\documentclass[aps,pre]{revtex4-1}

\usepackage[utf8]{inputenc}
\usepackage{color}
\usepackage{subfigure}
\usepackage{url}
\usepackage{graphicx}
\usepackage{amsmath}
\usepackage{amsthm}
\usepackage[plainpages=false]{hyperref}
\usepackage{algorithm}
\usepackage{algorithmic}

\usepackage{thmtools,thm-restate}
\declaretheorem{theorem}
\declaretheorem[sibling=theorem]{lemma}

\newtheorem{corollary}[theorem]{Corollary}

\theoremstyle{definition}

\newcommand{\figref}[1]{Figure \ref{fig:#1}}

\newcommand{\lemref}[1]{Lemma \ref{lemma:#1}}

\newcommand{\corref}[1]{Corollary \ref{cor#1}}

\newcommand{\theoref}[1]{Theorem \ref{theo:#1}}
\newcommand{\theorefX}[1]{\ref{theo:#1}}

\newcommand{\algref}[1]{Algorithm \ref{alg:#1}}
\renewcommand{\eqref}[1]{({\ref{eq:#1}})}

\newcommand{\lemlab}[1]{\label{lemma:#1}}

\newcommand{\corlab}[1]{\label{cor#1}}
\newcommand{\theolab}[1]{\label{theo:#1}}
\newcommand{\seclab}[1]{\label{sec:#1}}
\newcommand{\alglab}[1]{\label{alg:#1}}
\newcommand{\eqlab}[1]{\label{eq:#1}}

\newcommand{\iprod}[2]{\left\langle {#1},{#2}\right\rangle}

\newcommand{\pdot}{\dot{p}}

\begin{document}

\title{Anchored boundary conditions for locally isostatic networks}
\author{Louis Theran}
\email[Email: ]{louis.theran@aalto.fi}
\homepage[Web: ]{http://theran.lt}
\affiliation{Aalto Science Institute (AScI) and Department of Computer Science
(CS),
Aalto University, PO Box 15500, 00076 Aalto, Finland}

\author{Anthony Nixon}
\email[Email: ]{a.nixon@lancaster.ac.uk}
\homepage[Web: ]{http://www.lancaster.ac.uk/maths/about-us/people/anthony-nixon}
\affiliation{Department of Mathematics and Statistics, Lancaster University,
Lancaster LA1 4YF, England}

\author{Elissa Ross}
\email[Email: ]{elissa.ross@meshconsultants.ca}
\homepage[Web: ]{http://www.elissaross.ca}
\affiliation{MESH Consultants Inc., Fields Institute for Research in the
Mathematical Sciences, 222 College Street, Toronto, ON, M5T 3J1, Canada}

\author{Mahdi Sadjadi}
\email[Email: ]{ssadjadi@asu.edu}
\affiliation{Department of Physics, Arizona State University, Tempe, AZ 85287-1504,
USA}

\author{Brigitte Servatius}
\email[Email: ]{bservat@wpi.edu}
\homepage[Web: ]{http://users.wpi.edu/~bservat}
\affiliation{Department of Mathematical Sciences, Worcester Polytechnic
Institute, 100 Institute Road, Worcester, MA 01609, USA}

\author{M. F. Thorpe}
\email[Email: ]{mft@asu.edu}
\homepage[Web: ]{http://thorpe2.la.asu.edu/thorpe}
\affiliation{Department of Physics, Arizona State University, Tempe, AZ 85287-1504,
USA and Rudolf Peierls Centre for Theoretical Physics, University of Oxford, 1
Keble Rd, Oxford OX1 3NP, England}

\begin{abstract}
Finite pieces of locally isostatic networks have a large number of floppy modes because
of missing constraints at the surface. Here we show that by imposing suitable boundary
conditions at the surface, the network can be rendered {\it{effectively
isostatic}}. We refer to these as {\it{anchored boundary conditions}}. An
important example is formed by a two-dimensional
network of corner sharing triangles, which is the focus of this paper. Another
way of rendering such networks isostatic, is by adding an external wire along
which all unpinned vertices can slide ({\it{sliding boundary conditions}}). This
approach
also allows for the incorporation of  boundaries
associated with internal {\it{holes}} and complex sample geometries, which are
illustrated with examples. The recent synthesis of bilayers of vitreous silica
has provided impetus for this work. Experimental results from the imaging of
finite pieces at the atomic level needs such boundary conditions, if the observed
structure is to be computer-refined  so that the interior atoms have the perception
of being in an infinite isostatic environment.
\end{abstract}

\maketitle

\section{Introduction} \seclab{intro}
Boundary conditions are paramount in many areas of computer modeling in science.
At the atomic level, finite samples require appropriate boundary conditions in
order that atoms in the interior behave as if they were part of a larger or
infinite sample, or as closely to this as is possible. One example of this is
the calculation of the electronic properties of covalent materials where the
surface is terminated with H atoms so that all the chemical valency is
satisfied.  In this way the HOMO (highest occupied molecular orbital) and the
LUMO (lowest unoccupied molecular orbital) states inside the sample can be
obtained that are not very different from those expected in the bulk sample.
In materials science the electronic band structure of a sample of crystalline
Si could be obtained by determining the electronic properties of a finite
cluster terminated with H bonds at the surface.  In practice this is rarely
done, as it is more convenient to use periodic boundary conditions and hence
use Bloch's theorem, but this technique has been used recently in graphene
nanoribbons \cite{HPS07}.

For most samples, the nature of the boundary, fixed, free or periodic only alters
the properties of the sample by the ratio of the number of atoms on the surface
to those in the bulk. This ratio is $N^{-\frac{1}{d}}$ where $N$ is the number
of atoms ({\it{later referred to as vertices}}) and $d$ is the dimension.  Of
course this ratio goes to zero in the thermodynamic limit as the size of the
system $N\rightarrow\infty$ and leads to the important result that properties
become independent of boundary conditions for large enough systems.

Similar statements can be made about the mechanical and vibrational properties of systems
{\it{except}} for isostatic networks  that lie on the border of mechanical instability.
In this case the boundary conditions are important no matter how large $N$,
and special care must be taken with devising boundary conditions
so that the interior atoms behave as if they were part of an infinite sample,
in as much as this is possible \cite{T95,LKMSS15,EHKTH15}.

In Figure~\ref{fig:Figure1}, we show a part of a Scanning Probe Microscope (SPM) image
\cite{LBYSHSWSF12} of a bilayer of vitreous silica which has the chemical formula SiO$_2$. The
sample consists of an upper layer of tetrahedra with all the apexes pointing downwards where
they join a mirror image in the lower layer. In the figure we show the triangular faces of the
upper tetrahedra, which form rigid triangles with a (red) Si atom at the center and the (black)
O atoms at the vertices of the triangles which are freely jointed to a good approximation. We
refer to these networks as {\it{locally isostatic}} as the number of degrees of freedom of the
equilateral triangle in two dimensions is exactly balanced by the shared pinning constraints (2
at each of the 3 vertices, so that $3 -2 \times 3/2 =0$). While the $3$D bilayers are locally
isostatic, so too are the $2$D projections of corner-sharing triangles which are the focus of
this paper. We will use the {\it{Berlin A}} sample as the example throughout
\cite{WKST13,KSWT14} so that we can focus on this single geometry for pedagogical purposes.

Because experimental samples are always finite in extent and usually have irregular boundaries,
including internal regions that are either absent, or not imaged it is necessary to develop
appropriate boundary conditions. Note that the option of cutting a rectangular piece out of the
experimental image is not available because of the amorphous nature of the network, which means
that it is not possible for the left side to connect to the right side as with a regular
crystalline network. Even if this were possible, it would be unwise to discard experimental
data and hence loose information. In this paper we show how boundary conditions can be applied
to locally isostatic systems which are not periodic.

\begin{figure}[tb]
\includegraphics[width=7.3cm,angle=0]{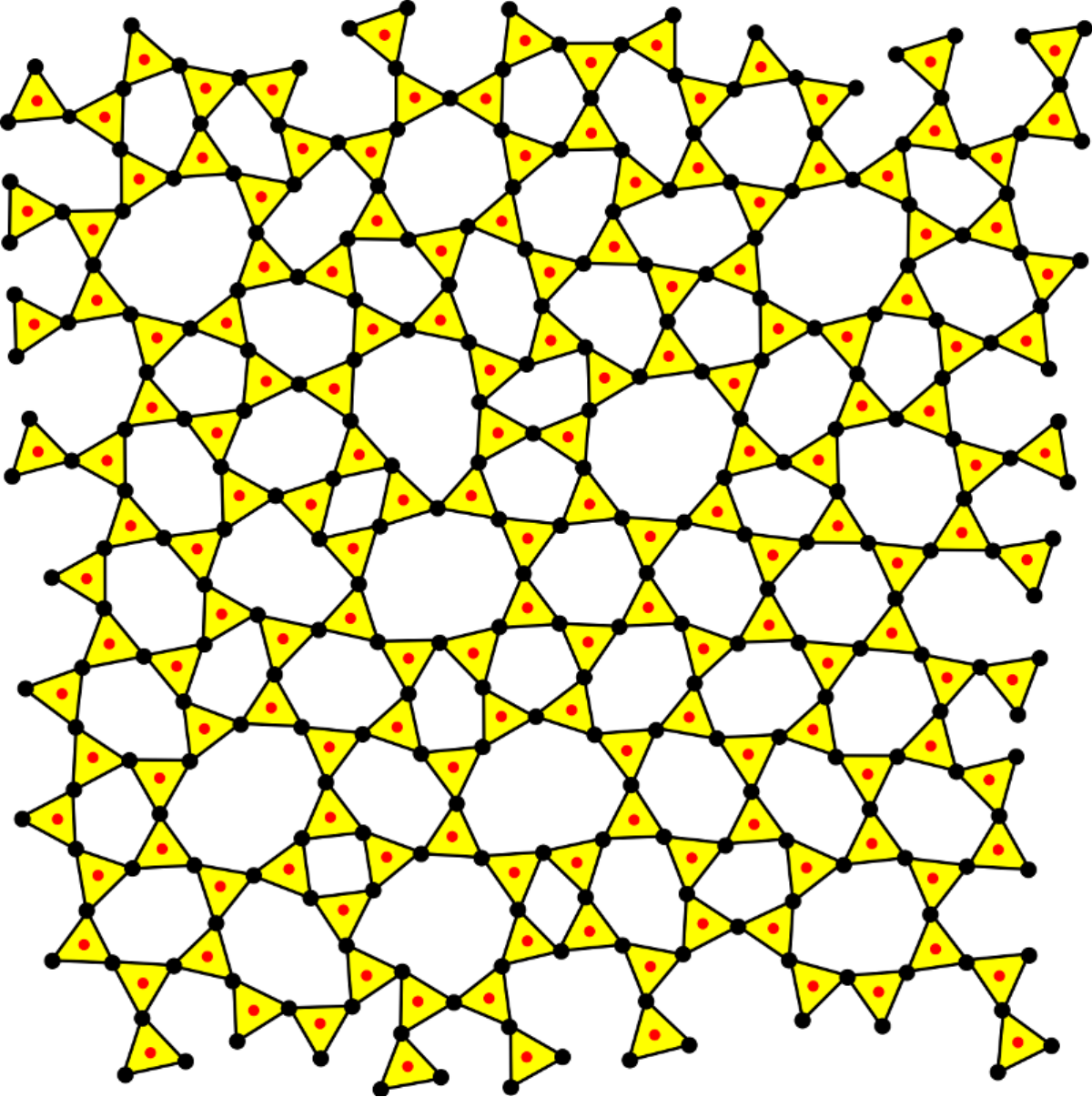}
\caption{Showing a piece of bilayer of vitreous silica imaged in SPM (Scanning Probe Microscope)
\cite{LBYSHSWSF12} to show the Si atoms as red discs and the O atoms as black discs. The local
covalent bonding leads to the yellow almost-equilateral triangles that are freely jointed, which we
will refer to as {\it{pinned}}. The triangles at the surface have either one or two vertices unpinned.}
\label{fig:Figure1}
\end{figure}

In this paper, we show rigorously that there are various ways to add back the exact
number of missing constraints at the surface, in a way that they are sufficiently
uniformly distributed around the boundary that the network is guaranteed to be
isostatic everywhere.  There is some limited freedom in the precise way these boundary
conditions are implemented, and the boundary can be general enough to include internal
holes. The proof techniques used here involves showing that \emph{all} subgraphs
have insufficient edge density for redundancy to occur \cite{L70}.  In the appendix,
we give an algorithmic desctiption of our boundary conditions and discuss
in detail how to ensure the resulting boundary is sufficiently generic.

Using the pebble game \cite{JT95,JT96}, we verified on a number of samples
that anchored boundary conditions in which alternating free vertices are
pinned results in a global isostatic state.  The pebble game is an integer
algorithm,  based on Laman's theorem \cite{L70}, which for a particular network
performs a rigid region decomposition, which involves finding  the rigid regions,
the hinges between them, and the number  of floppy (zero-frequency) modes.
We have used it to confirm that the locally isostatic samples  such as that in
this paper are isostatic overall with anchored boundary conditions. The results
of  this paper imply that, under a relatively mild connectivity hypothesis, this
procedure is provably  correct, and thus, relatively robust.  Additionally, the
necessity of running the pebble game for each individual case is avoided.

\figref{Figure2} shows \emph{sliding} boundary condition \cite{ST10}. These
make use of a different, simpler kind of geometric constraint at each unpinned
surface site.  The global effect on the network's degrees of freedom is like that
of the  anchored boundary conditions, and this setup is computationally reasonable.
At the same time, the proofs for this case are simpler, and generalize  more
easily to handle situations such as holes in the sample.

In Figure~\ref{fig:Figure3}, we show the anchored boundary conditions. We have
trimmed off the surface triangles in Figure~\ref{fig:Figure1} that are only
pinned at one vertex. This makes for a a more compact structure whose properties
are more likely to mimic those of a larger sample, and makes our mathematical statements
easier to formulate. In addition we have had to remove the 3 purple triangles at
the lower right hand side in order to get an even number of unpinned surface sites.
When the network is embedded in the plane, this is possible, except for very
degenerate samples (see \figref{Figure3}).

\begin{figure}[tb]
\includegraphics[width=7.3cm,angle=0]{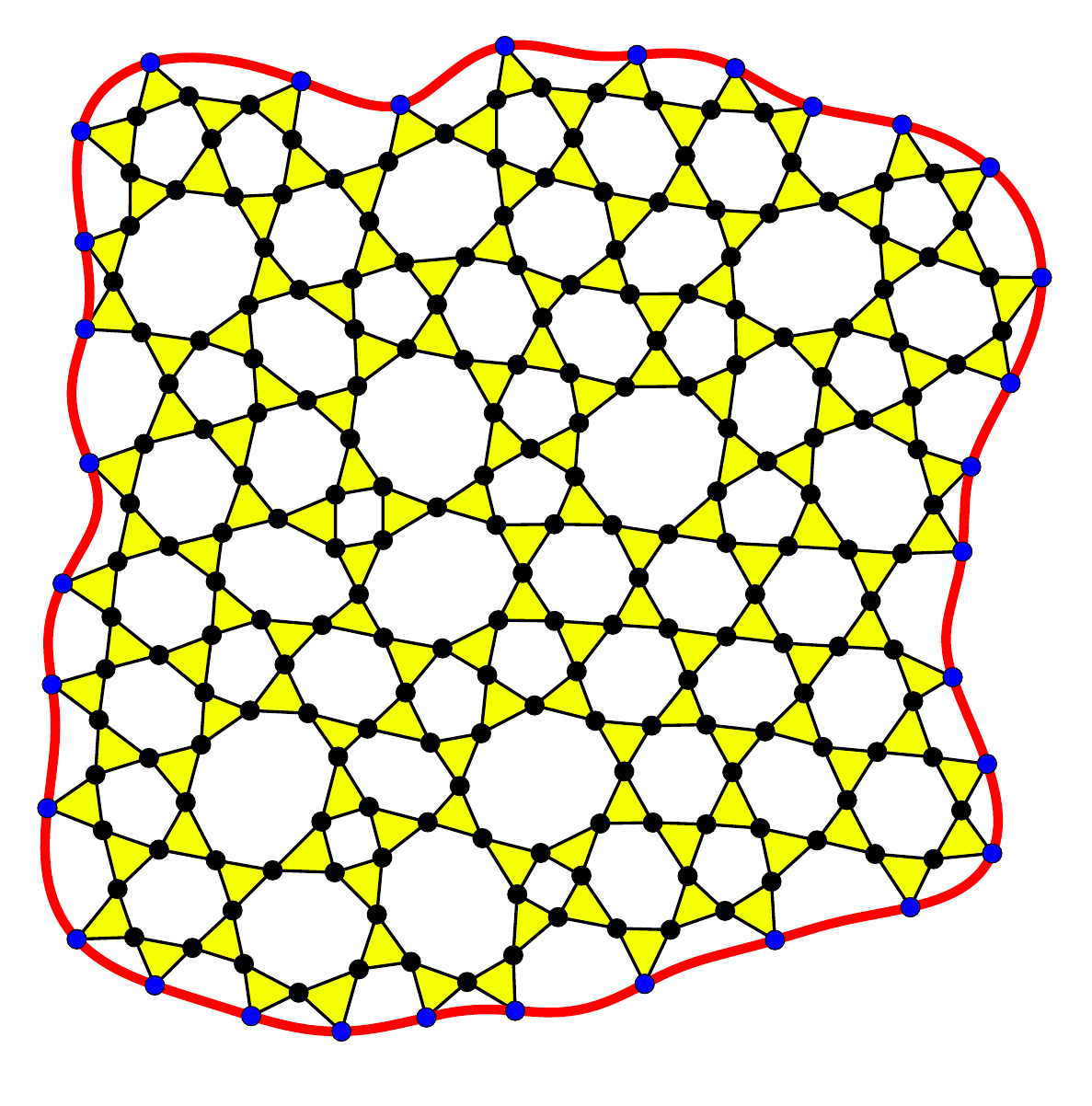}
\caption{Illustrating sliding boundary conditions, used for a piece of the sample shown
in Figure~\ref{fig:Figure1}. The boundary sites are shown as blue discs and the
3 purple triangles at the lower left Figure~\ref{fig:Figure3} have been
removed. The red Si atoms at the centers of the triangles in
Figure~\ref{fig:Figure1} have also been removed for clarity.
The boundary is formed as a smooth analytic curve by using a Fourier series with 16
sine and 16 cosines terms to match the number of surface vertices, where the
center for the radius $r(\theta)$ is placed at the centroid of the 32 boundary
vertices \cite{sadjadi2017refining}. Note that sliding boundary conditions do not require
an even number of boundary sites.}
\label{fig:Figure2}
\end{figure}

\begin{figure}[tb]
\includegraphics[width=7.3cm,angle=0]{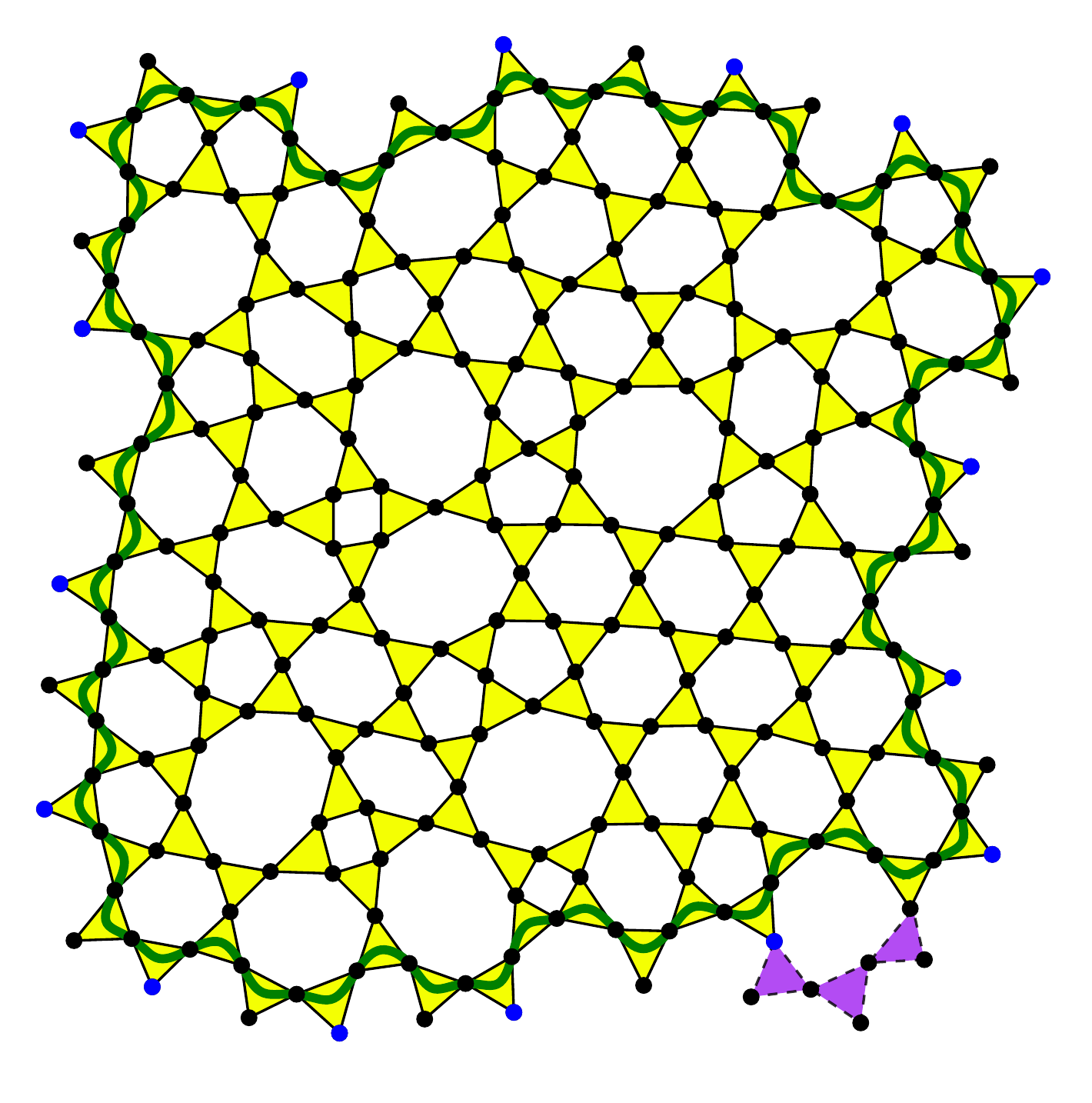}
\caption{Illustrating the anchored boundary conditions used for the sample
shown in Figure~\ref{fig:Figure1}.  The alternating anchored sites on the
boundary are shown as blue discs and the 3 purple triangles at the lower right
are removed to give an even number of unpinned surface sites. The red Si atoms
at the centers of the triangles in Figure~\ref{fig:Figure1} have been suppressed
for clarity.}
\label{fig:Figure3}
\end{figure}

\begin{figure}[tb]
\includegraphics[width=7.3cm,angle=0]{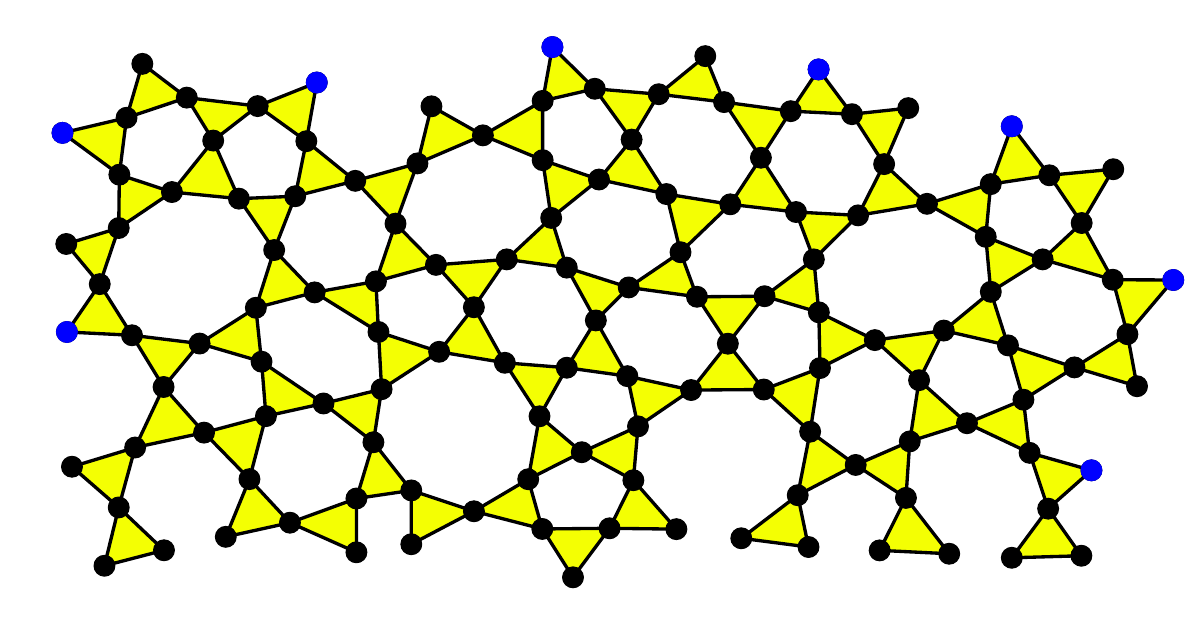}
\caption{Showing a typical subgraph from Figure \ref{fig:Figure3} used in the
proof
that there are no rigid subgraphs larger than a single triangle. (See
\lemref{rigcomps-2}.)}
\label{fig:Figure4}
\end{figure}

\begin{figure}[tb]
\includegraphics[width=7.3cm,angle=0]{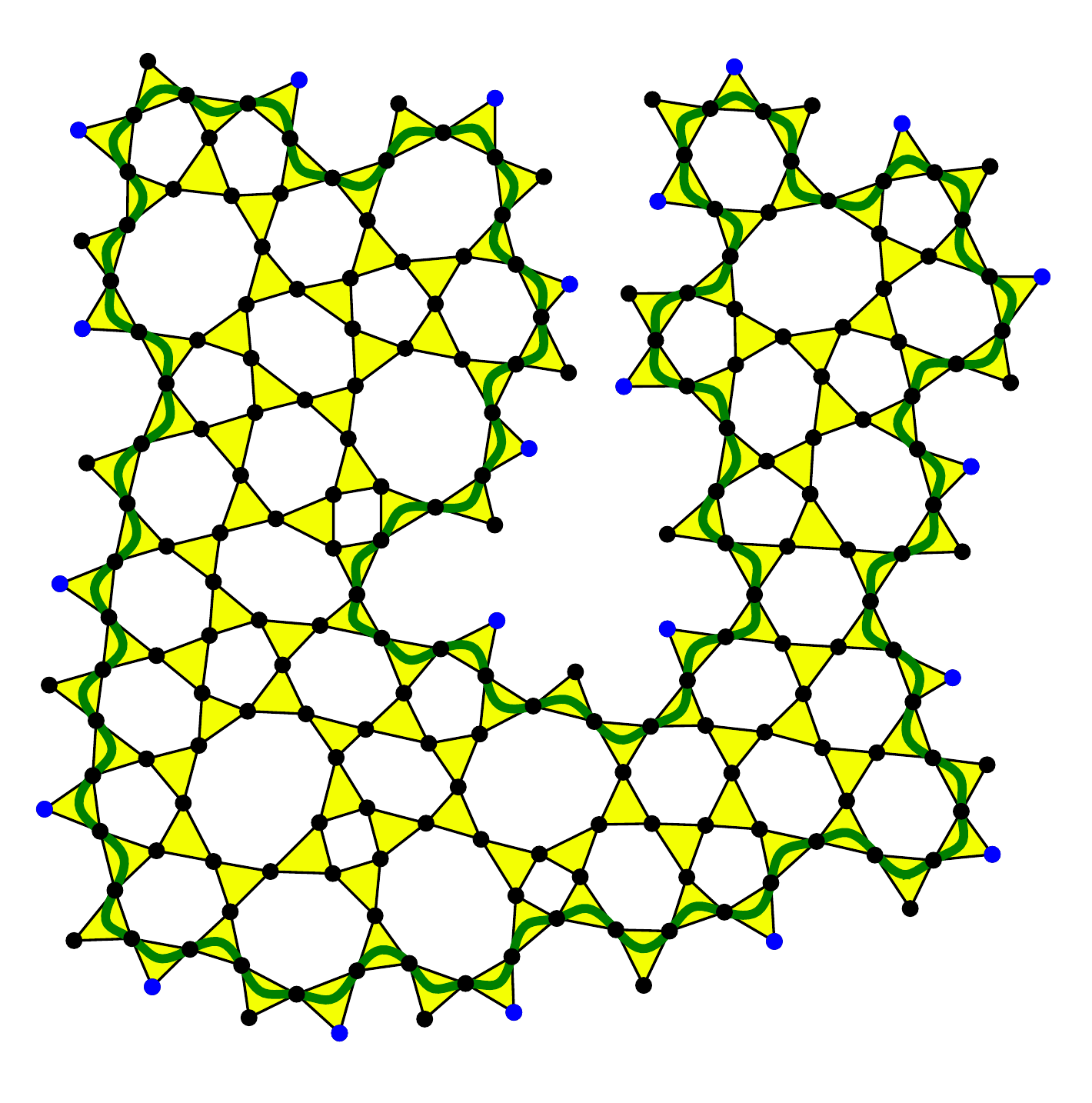}
\hfill
\includegraphics[width=7.3cm,angle=0]{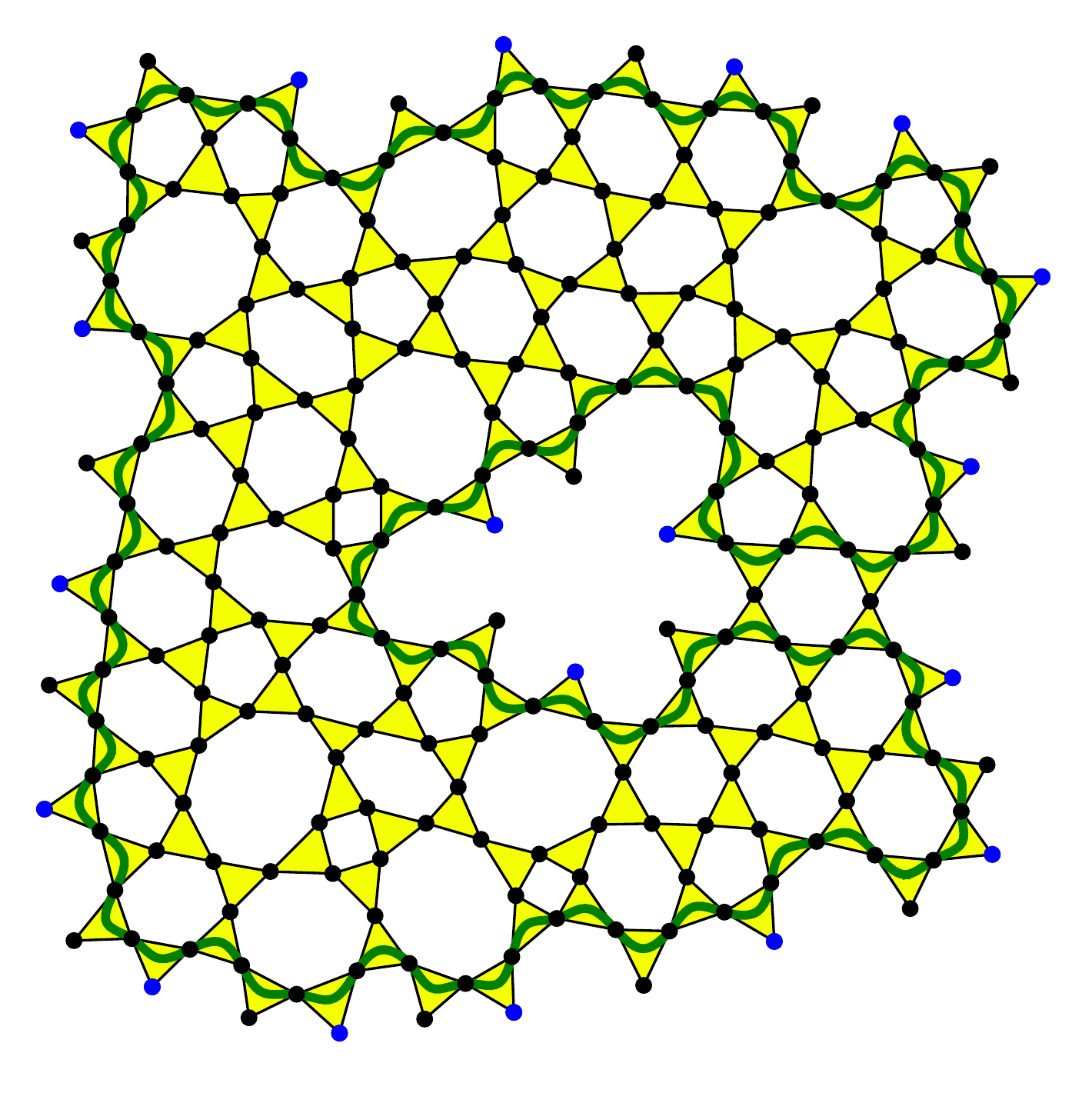}
\caption{Illustrating two, at first sight, more complex anchored boundary
conditions that by our results can be used for the sample shown in
Figure~\ref{fig:Figure2}, with the 3 purple triangles at the lower left are
removed to give an even number of unpinned surface sites.
The anchored sites are shown as blue discs, with an even number of surface sites
in both graphs. The graph at the right has an even number of surface sites in
{\it{both}} the outer and inner boundary.  The red Si atoms at the centers of
the triangles have been suppressed for clarity. The green line goes through the
boundary triangles.}
\label{fig:Figure5ab}
\end{figure}

\begin{figure}[tb]
\includegraphics[width=7.3cm,angle=0]{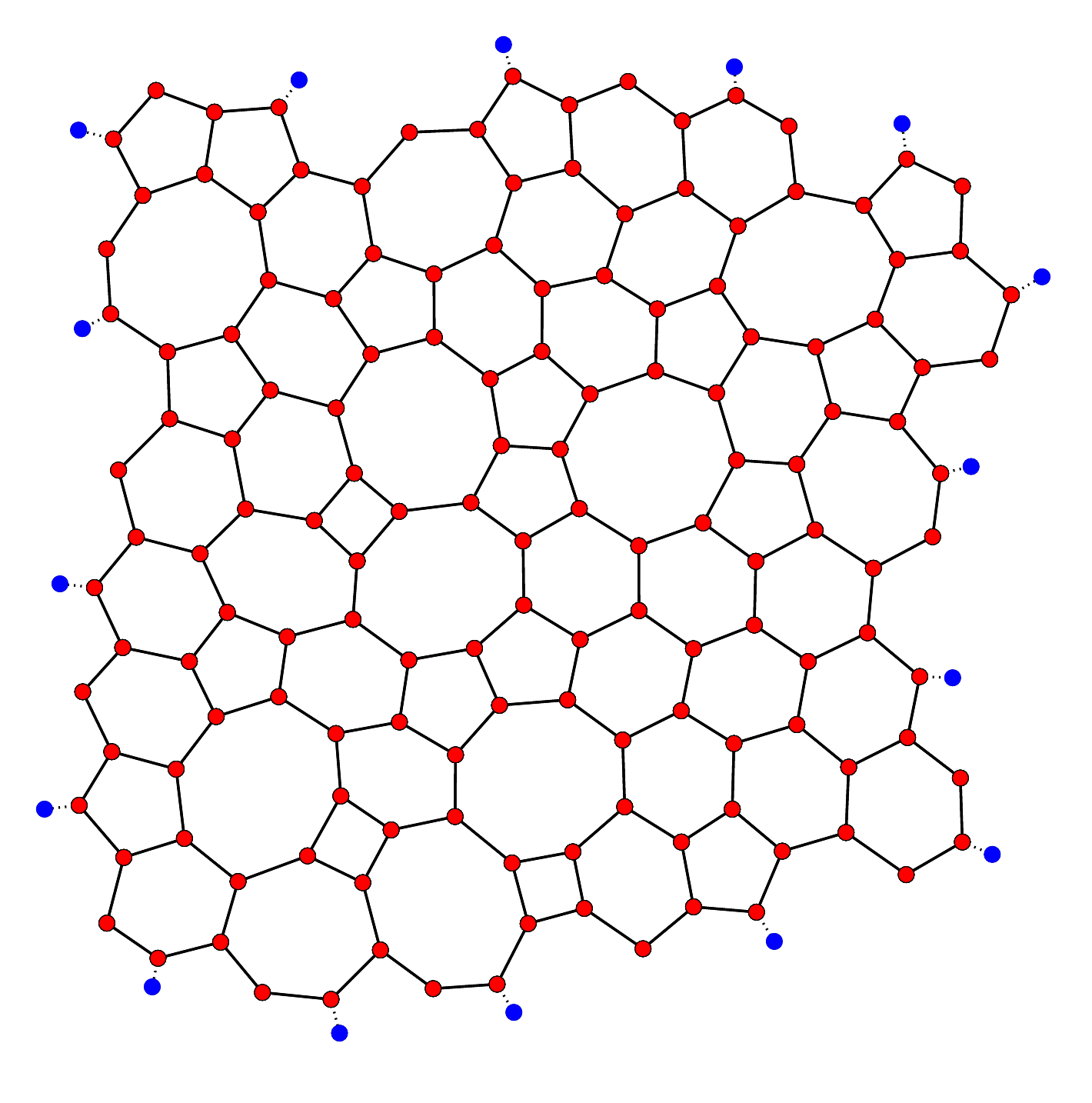}
\caption{The triangle ring network, complementary to that in
Figure~\ref{fig:Figure3}, where the Si atoms, shown as
red discs, at the center of each triangle are emphasized in this
three-coordinated network. Dashed edges are
shown connecting to the anchored sites.}
\label{fig:Figure6}
\end{figure}

\section{Combinatorial anchoring}\seclab{comb}
Intuitively, the internal degrees of freedom of systems like the ones
in Figures \ref{fig:Figure1} and \ref{fig:Figure3} correspond
to the corners of trianges that are not shared.  This is, in
essence, the content of \lemref{indep} proved below.  Proving
\lemref{indep} requires ruling out the appearance of \emph{additional}
degrees of freedom that could arise from \emph{sub-structures}
that contain more constraints than degrees of freedom.

The essential idea behind combinatorial rigidity\footnote{See, e.g., the
monograph by Graver, et al. \cite{GSS93} for an introduction.} is that \emph{generically}
all geometric constraints are visible from the topology of the structure,
as typified by Laman's \cite{L70} striking result showing the sufficiency
of Maxwell counting \cite{M64} in dimension $2$.  Genericity means,
roughly, that there is no special geometry present; in particular, generic
instances of any topology are dense in the set of all instances.

In what follows, \emph{we will be assuming genericity}, and then use
results similar to Laman's, in that they are based on an appropriate
variation of Maxwell counting.  Our proofs have a graph-theoretic
flavor, which relate certain hypotheses about connectivity\footnote{To make
this paper somewhat self-contained, we will briefly explain the concepts
we use.  Our terminology is standard, and can be found in, e.g., the
textbook by Bondy and Murty \cite{BM08}.} to hereditary Maxwell-type
counts.

\subsection{Triangle ring networks}
We will model the flexibility in the upper layer of vitrious silica bilayers
as systems of $2$D triangles, pinned togehter at the corners.  The joints at
the corners are allowed to rotate freely.  A triangle ring network is \emph{rigid}
if the only available motions preserving  triangle shapes and the network's
connectivity are rigid body motions; it is  \emph{isostatic} if it is rigid,
but ceases to be so once any joint is removed. These are an examples of body-pin
networks\footnote{Since only two triangles are pinned
together at any point, we are dealing with the $2$-dimensional specialization
of body-hinge frameworks first studied by Tay \cite{T89} and Whiteley \cite{W88}
in general dimensions.  In $2$D, there is a richer combinatorial
theory of ``body-multipin'' structures, introduced by Whiteley \cite{W89}.
See Jackson and Jordán \cite{JJ08} and the references therein for an overview
of the area.} from rigidity theory.

The combinatorial model is a graph $G$ that has one vertex for each triangle
and an edge between two triangles if they share a corner (Figure \ref{fig:Figure6}).
Since we are assuming
genericity, we will identify a geometric realization with the graph $G$ from
now on.  In what follows, we are interested in a particular class of graphs $G$,
which we call \emph{triangle ring networks}.  The definition of a triangle
ring is as follows: (a) $G$ has only vertices of degree $2$ and $3$;
$G$ is $2$-connected\footnote{This means that to disconnect $G$,
we need to remove at least $2$ vertices.}; (b) there is a simple cycle
$C$ in $G$ that contains all the degree $2$ vertices, and there are
at least $3$ degree $2$ vertices; (c) any edge cut set\footnote{This is a
inclusion-wise minimal set of edges that, when removed from $G$,
results in a graph  $2$ connected components.}
in $G$ that disconnects a subgraph containing only degree $3$
vertices has size at least $3$.

To set up some terminology, we call the degree $2$ vertices \emph{boundary
vertices} and the degrees $3$ vertices \emph{interior vertices}.  A
subgraph spanning only interior vertices is an \emph{interior subgraph}.

The reader will want to keep in mind the specific case in which $G$ is
planar with a given topological embedding and $C$ is the outer face, as
is the case in our figures. This means that subgraphs strictly interior to
the outer face have only interior vertices, which explains our terminology.
However, as we will discuss in detail later, the setup is very general.
If the sample has holes, $C$ can leave the outer boundary and return to it:
provided that it is simple, all the results here still apply.

A theorem of Tay--Whiteley \cite{W88,T89} gives the degree of
freedom counts for networks of $2$-dimensional bodies pinned together.
Generically, there are no stressed subgraphs in such a network,
with graph $G$, of $v$ bodies and $e$ pins if and only if
\begin{eqnarray}\eqlab{tw-count}
2e'\le 3v' - 3 & \qquad & \text{for all subgraphs $G'\subset G$.}
\end{eqnarray}
where $v'$ and $e'$ are the number of vertices and edges of the subgraph.
If \eqref{tw-count} holds for all subgraphs, the rigid subgraphs are
all isostatic, and they are the subgraphs where
\eqref{tw-count} holds  with equality.

\begin{lemma}\lemlab{indep}
Any triangle ring network $G$ satisfies \eqref{tw-count}.
\end{lemma}
\begin{proof}
Suppose the contrary.  Then there is a vertex-induced subgraph $T$ on $v'$
vertices
that violates \eqref{tw-count}. If $T$ contains a vertex $v$ of degree 1 then
$T-v$
also violates \eqref{tw-count} so we may assume that $T$ has minimum degree 2.
In this case, $T$ has at most $2$ vertices of degree $2$, since it has
maximum degree $3$. In particular, $T$ may be disconnected from $G$ by removing
at most $2$ edges.  If $T$ is an interior subgraph, we get a contradiction right
away.
Alternatively, at least one of the degree $2$ vertices in $T$ is degree $2$ in
$G$, and so on $C$.  If exactly one is, then $G$ is not $2$-connected.  If both
are,
then $T=G$ and there are only $2$ boundary vertices.  Either case is a
contradiction.
\end{proof}

\begin{corollary}\corlab{rigcomps}
The rigid subgraphs of a triangle ring network $G$ are the
subgraphs containing exactly $3$ vertices of degree $2$ and every other
vertex has degree $3$. Moreover, any proper rigid subgraph contains at most
one boundary vertex of $G$.
\end{corollary}
\begin{proof}
The first statement is straightforward.  The second follows from observing that
if a rigid subgraph $T$ has two vertices on the boundary of $G$, then $G$
cannot be $2$-connected, since all the edges detaching $T$ from $G$ are
incident on a single vertex.
\end{proof}
When $G$ is planar, these rigid subgraphs are regions cut out by cycles of
length $3$
in the Poincaré dual.  More generally in the planar case, subgraphs
corresponding to regions that are smaller triangle ring networks
with $t$ degree $2$ vertices have $t$ degrees of freedom.

\subsection{Anchoring with sliders}
Now we can consider our first anchoring model, which uses
\emph{slider pinning} \cite{ST10}.
A \emph{slider} constrains the motion of a point to
remain on a fixed line, rigidly attached to the plane.
When we talk about attaching sliders to a vertex of the graph,
we choose a point on the corresponding triangle, and constrain its motion by
the slider.  In the results used below, this point should be chosen generically;
for example the theory does not apply if the slider is attached at a pinned
corner shared by two of the triangles.  Since we are only attaching sliders
to triangles corresponding to degree $2$ vertices in $G$, we may always attach
sliders at an unpinned triangle corner.

The notion of rigidity for networks of bodies with sliders is that of
being \emph{pinned}: the system is completely immobilized.\footnote{Rigid body
motions are
not ``trivial'', because slider constraints are not preserved by them.}  A
network with sliders is \emph{pinned-isostatic} if it is pinned, but ceases
to be so if any pin or slider is removed.

The equivalent of the White-Whiteley counts in the presence of sliders is a
theorem of Katoh and Tanigawa \cite{KT13}, which says that a generic
slider-pinned
body-pin network $G$ is independent if and only if the body-pin graph
satisfies \eqref{tw-count} and
\begin{eqnarray}\eqlab{pin-count}
2e' + s' \le 3v' & \qquad & \text{for all subgraphs $G'\subset G$,}
\end{eqnarray}
where $s'$ is the number of sliders on vertices of  $G'$.  Here is our first
anchoring procedure.
\begin{theorem}\theolab{slider-pinning}
Adding one slider to each degree $2$ boundary vertex of a triangle ring network
$G$ gives a pinned-isostatic network.
\end{theorem}
\begin{proof}
Let $T$ be an arbitrary subgraph with $v'$ vertices and $v''$ vertices
of degree at most $2$. That \eqref{tw-count} holds is \lemref{indep}.
The fact that the only vertices of $T$ which get a slider are vertices with
degree 2 in $G$ implies that \eqref{pin-count} is also satisfied, and, by
construction $2e + s = 3v$.
\end{proof}
We may think of this anchoring as rigidly attaching a rigid
wire to the plane then constraining the boundary vertices to
move on it.  Provided that the wire's path is smooth and
sufficiently non-degenerate, this is equivalent, for
analyzing infinitesimal motions,
to putting the sliders in the direction of the
tangent vector at each boundary vertex. See also
Figure~\ref{fig:Figure2}.

\begin{figure}[tb]
\includegraphics[width=7.3cm,angle=0]{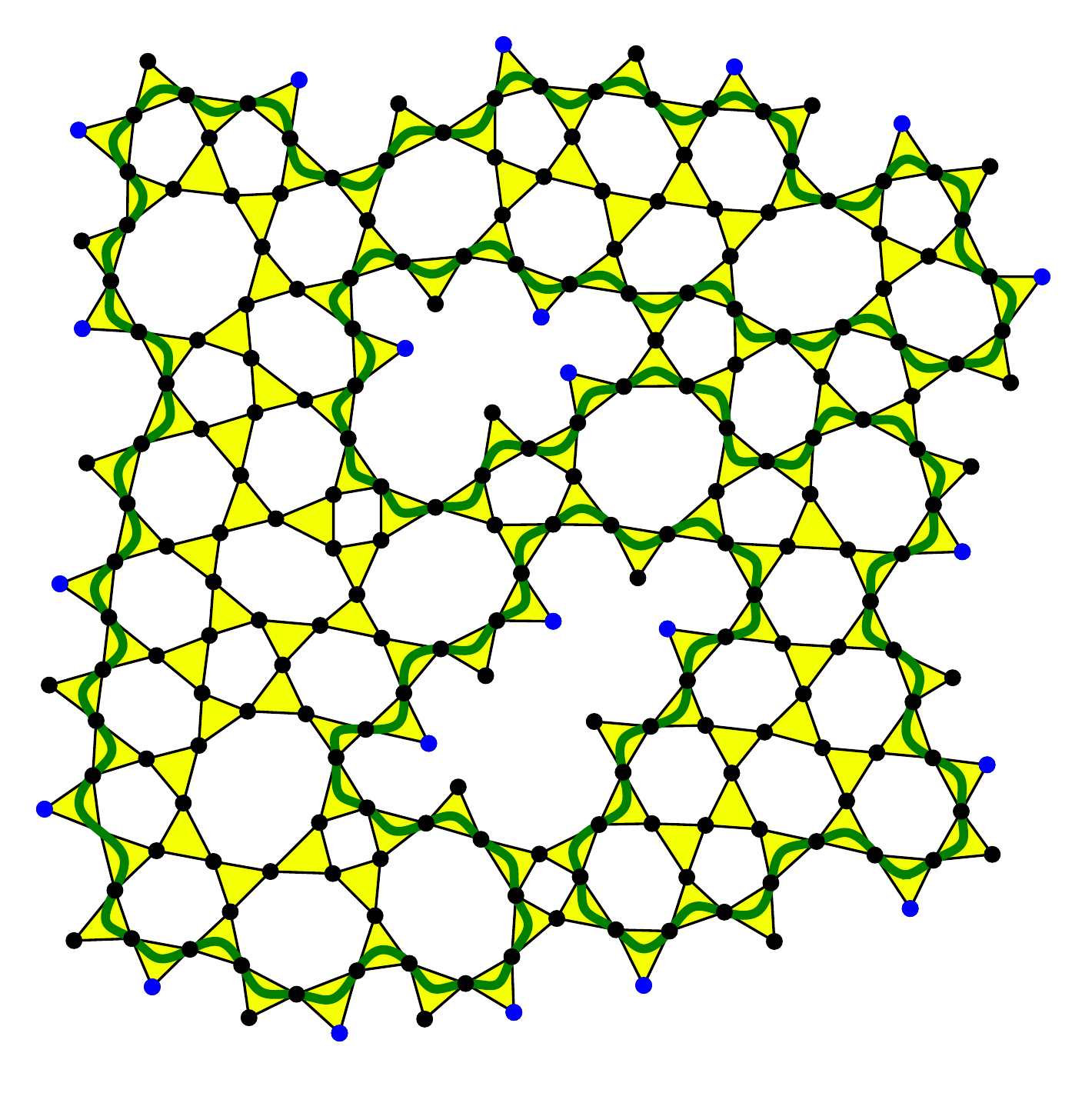}
\caption{Illustrating even more complex boundaries, developed from the sample
shown in Figure~\ref{fig:Figure3} by removing triangles to form two internal
{\it{holes}}. The boundary sites are shown as blue discs and the 3 purple
triangles at the lower left Figure~\ref{fig:Figure3} have been removed. The
red Si atoms at the centers of the triangles in Figure~\ref{fig:Figure1} have
also been removed for clarity. The green line forms a continuous
{\it{boundary}} which goes through all the surface sites which must be an even
number.  The anchored (blue) sites then alternate with the unpinned sites on
the green boundary curve which has to cross the bulk sample in two places to
reach the two internal holes.  Here there are 32 boundary sites, 5 boundary
sites in the upper hole and 7 in the lower hole, giving a total even number of
44 boundary sites. Where these crossings take place is arbitrary, but it is
important that the anchored and unpinned surface sites alternate along whatever
(green) boundary line is drawn.}
\label{fig:Figure7}
\end{figure}

\subsection{Anchoring with immobilized triangle corners}
Next, we consider anchoring $G$ by immobilizing (pinning) some points completely.
Combinatorially, we model pinning a triangle's corner by adding two sliders
through it.  Since we are still using sliders, the definitions of
pinned and pinned-isostatic are the same as in the previous section.

The analogue for \eqref{pin-count} when we add sliders in groups of $2$ is:
\begin{eqnarray}\eqlab{thumbtack-count}
2e' + 2s' \le 3v' & \qquad & \text{for all subgraphs $G'\subset G$,}
\end{eqnarray}
where $s'$ is the number of immobilized corners.
\begin{theorem}\theolab{thumbtack-pinning}
Let $G$ be a triangle ring network with an even number $t$ of
degree $2$ vertices on $C$.
Then, following $C$ in cyclic order, pinning every other boundary
vertex that is encountered results in a pinned-isostatic network.
\end{theorem}
\begin{proof}
Let $T$ be an arbitrary subgraph of $G$.  If at most one of the vertices of $T$
are pinned, there is nothing to do.  For the moment, suppose that no
vertex of degree $1$ in $T$ is pinned.  Let $t$ be the number of pinned
vertices in $T$.

We will show that for each of the $t$ pinned vertices, there is a distinct
unpinned vertex of degree $1$ or $2$ in $t$. This implies that $2e' \le 3v' - 2t$ in
$T$, at which point we know \eqref{thumbtack-count} holds for $T$.

To prove the claim, let $v$ be a pinned vertex of $T$.  Traverse the boundary
cycle $C$ from $v$.  Let $w$ be the next pinned vertex of $T$ that is encountered.
If the chain from $v$ to $w$ along $C$ is in $T$, the
alternating pattern provides an unpinned degree $2$ vertex that is degree $2$ in $G$.
Otherwise, this path leaves $T$, which can only happen at a vertex with degree $1$ or $2$ in
$T$.  Continuing the process until we return to $v$, produces at least $t$
distinct unpinned degree $2$ vertices, since each step considers a
disjoint set of vertices of $C$.

Now assume that $T$ does have a pinned vertex $v$ of degree $1$.  The theorem will
follow if \eqref{thumbtack-count} holds strictly for $T-v$.  Let $w$ and $x$ be the
pinned vertices in $T$ immediately preceding and following $v$.  The argument
above shows that there are at least $2$ unpinned degree $1$ or $2$ vertices
in $T$ on the path in $C$ between $w$ and $x$ on $C$.  Since these are in $T-v$,
we are done.
\end{proof}

When there are an odd number of boundary vertices in $G$, \theoref{thumbtack-pinning}
does not apply.  This next lemma gives a simple reduction in many cases of
interest.
\begin{theorem}\theolab{odd}
Let $G$ be a planar triangle ring network, with $C$ the outer face.
Suppose that there are an odd number $t$ of boundary vertices.
If $G$ is not a single cycle, then it is possible to obtain a
network with an even number of boundary vertices by removing the
intersection of a facial cycle of $G$ with $C$, unless $G=C$.
\end{theorem}
\begin{proof}[Proof sketch]
The connectivity requirements for a triangle ring network,
combined with planarity of $G$ imply that the intersection of $C$ and
any facial cycle $D$ of $G$ is a single chain.  Every boundary vertex is in the
interior of such a chain, so some facial cycle $D$ contributes an odd number of
boundary vertices.  Removing the edges in $D\cap C$ changes the parity of the
number of boundary vertices.
\end{proof}

\subsection{Anchoring with additional bars}
So far, we have worked with networks of triangles pinned together.  Now we augment
the model to also include bars between pairs of the triangles.  We will always
take the endpoints of the bars to be free corners of triangles that are boundary
vertices in the underlying network $G$. Combinatorially we model this by a graph
$H$ on the same vertex set as $G$, with an edge for each bar between a pair of
bodies. In this case, the Tay--Whiteley count becomes:
\begin{eqnarray}\eqlab{bar-pin-body-count}
2e' + b' \le 3v' - 3 & \qquad & \text{for all subgraphs $G'\subset G$.}
\end{eqnarray}
where $e'$ is the number of edges in $G'$ and $b'$ is the number of
edges in $H$ spanned by the vertices of $G'$.  The anchoring
procedures with sliders or immobilized vertices have analogues
in terms of adding bars to create an isostatic network. These boundary
conditions are illustrated on the right hand side of Figure~\ref{fig:Figure8ab}.
Also shown in Figure~\ref{fig:Figure8ab} in the left panel is a triangular
scheme involving alternating unpinned surface sites, that is equivalent to
anchoring. In both cases shown here the sample is free to rotate with respect to
the page.

\begin{figure}[tb]
\includegraphics[width=7.3cm,angle=0]{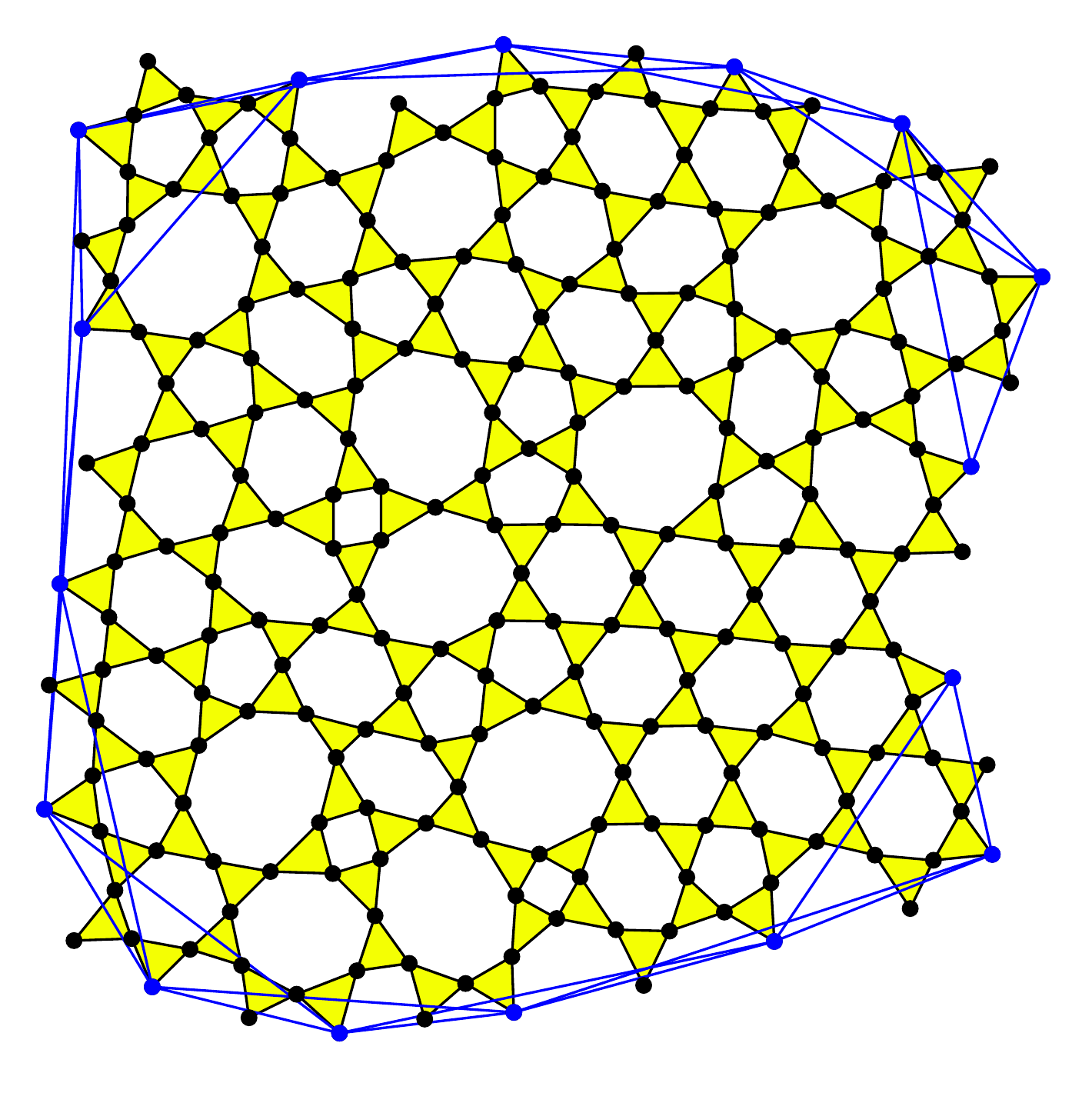}
\hfill
\includegraphics[width=7.3cm,angle=0]{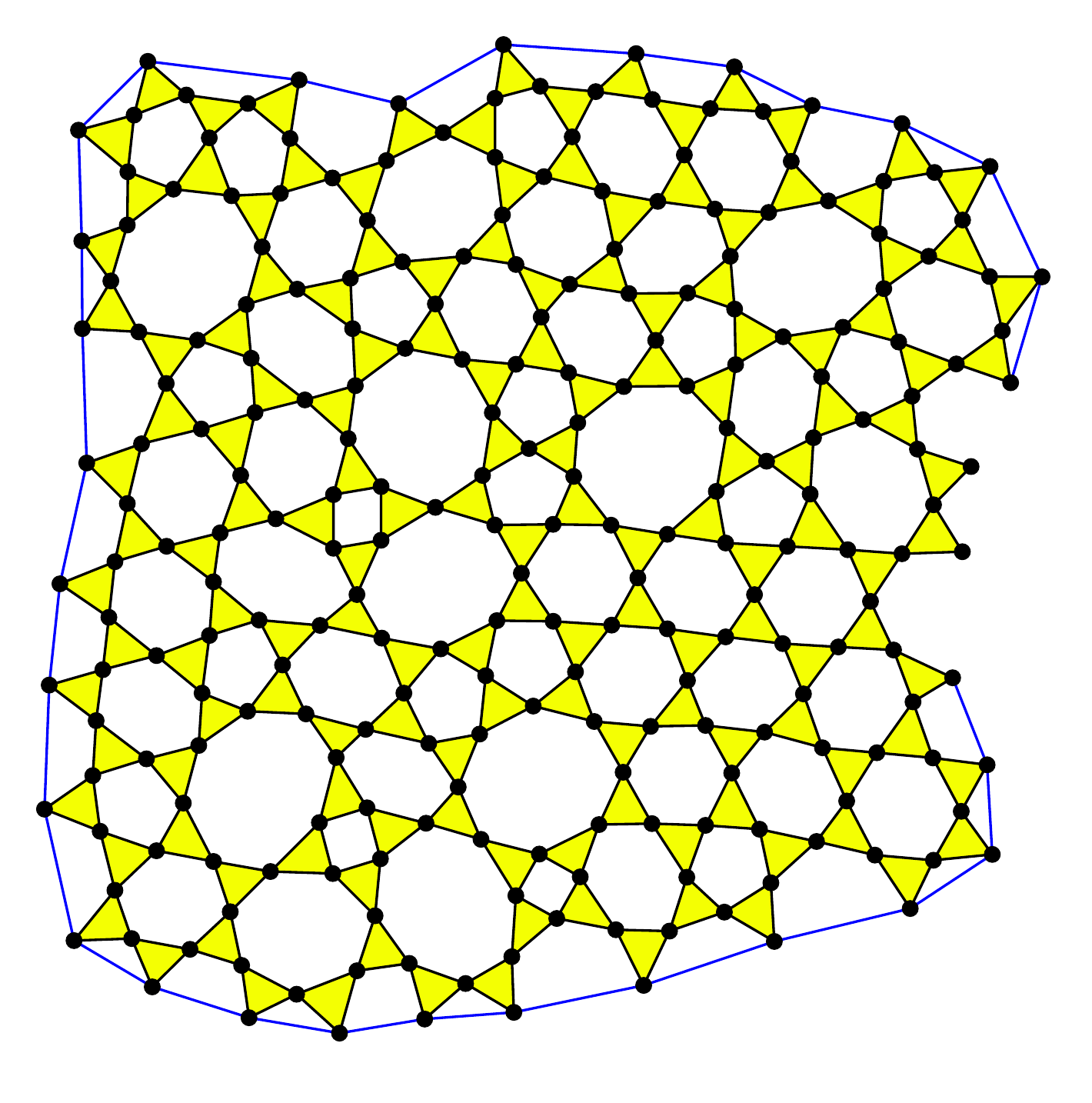}
\caption{Illustrating two additional boundary conditions used for the sample shown in
Figure~\ref{fig:Figure3}, with the 3 purple triangles at the lower left
removed to give an even number of unpinned surface sites.
On the left, alternating surface sites are connected to one another through
triangulation of first and second neighbors, with the last three connections not
needed (these would lead to redundancy). Hence there are three additional
macroscopic motions when compared to Figure~\ref{fig:Figure3} which can be
considered as being pinned to the page rather than to the {\it{internal frame}}
shown by blue straight lines. On the right we illustrate anchoring with
additional bars which connect all unpinned surface sites, except again three
are absent, to avoid redundancy, and to give the three additional macroscopic
motions when compared to Figure~\ref{fig:Figure3}}
\label{fig:Figure8ab}
\end{figure}

\begin{theorem}\theolab{bars-1}
If $G$ has boundary vertices $v_1,\ldots, v_t$, we obtain an
isostatic framework by taking the edges of $H$ to be
$v_1v_2, v_2v_3, \ldots, v_{t-3}v_{t-2}$.
\end{theorem}
\begin{proof}
Consider the $t-3$ new bars.  By construction and \lemref{indep} we have
$2e + t-3 =3v-t+t-3= 3v -3$.  \corref{rigcomps}
and the connectivity hypotheses imply that no rigid subgraph of $G$
has more than $1$ of its $3$ degree $2$ vertices on the boundary
of $G$.  This shows that no rigid subgraph of $G$ has a bar added to
it.
\end{proof}

\begin{theorem}\theolab{bars-2}
If $G$ has $t$ boundary vertices and $t$ is even, then taking
$H$ to be any isostatic bar-joint network with vertex set
consisting of $t/2$ boundary vertices chosen in an alternating
pattern around $C$ results in an isostatic network.
\end{theorem}
A triangulated $t/2$-gon is a simple choice for $H$.
\begin{proof}[Proof sketch]
By \lemref{indep}, we are adding enough bars to remove
all the internal degrees of freedom.  The desired statement
then follows from \theoref{thumbtack-pinning} by observing that
pinning down the boundary vertices is equivalent, geometrically,
to pinning down $H$ and then identifying the boundary vertices of $G$ to
the vertices of $H$.
\end{proof}

A result of White and Whiteley\cite{WW83} on ``tie downs'', then gives:
\begin{corollary}
In the situation of Theorems \theorefX{bars-1} and \theorefX{bars-2},
adding \emph{any} $3$ sliders results in a pinned-isostatic network.
\end{corollary}

\subsection{Stressed regions}
So far, we have shown how to render a floppy triangle ring
network isostatic or pinned-isostatic.  It is interesting to know
when adding a single extra bar or slider results in a network that
is stressed over all its members.  This is a somewhat subtle question
when adding bars or immobilizing vertices, but it has a simple answer
for the sliding boundary conditions.

We say that a triangle ring network is \emph{irreducible} if: (a)
every minimal $2$ edge cut set either detaches a single
vertex from $G$ or both remaining components contain more than one
boundary vertex of $G$; (b) every minimal $3$ edge cut set disconnects one
vertex from $G$.

\begin{lemma}\lemlab{rigcomps-2}
A triangle ring network $G$ has no proper rigid subgraphs if and only
if $G$ is irreducible.
\end{lemma}
\begin{proof}
Recall, from \corref{rigcomps}, that a proper rigid subgraph $T$ of $G$
has exactly $3$ vertices of degree $2$ and the rest degree $3$.  Thus, $T$
can be disconnected from $G$ by a cut set of size $2$ or $3$.

In the former case, \corref{rigcomps} implies that exactly one of the
degree $2$ vertices in $T$ is a boundary vertex of $G$.  This means
that $T$ witnesses the failure of (a), and $G$ is not irreducible.
Conversely, (a) implies that, for a $2$ edge cut set not disconnecting
one vertex, either side is either a chain of boundary vertices or has
at least $4$ vertices of degree $2$.

Finally, observe that cut sets of size $3$ are minimal if and only if they disconnect
an interior subgraph on one side.  \corref{rigcomps} then implies that there
is a proper rigid component that is an interior subgraph of $G$
if and only if (b) fails.
\end{proof}

\begin{theorem}\theolab{stresses}
Let $G$ be a triangle ring network anchored using the procedure
of \theoref{slider-pinning}.  Adding \emph{any} bar or slider to
$G$ results in a network with all its members stressed if and
only if $G$ is irreducible.
\end{theorem}
\begin{proof}
First consider adding a slider.  Because $G$ is pinned-isostatic,
the slider creates a unique stressed subgraph $T$.  A result of
Streinu-Theran\cite{ST10} implies that $T$ must have been
fully pinned in $G$.  Since any proper subgraph has an unpinned
vertex of degree $1$ or $2$, \eqref{pin-count} holds strictly.  Thus,
the stressed graph is all of $G$.\footnote{It is worth noting that,
so far, irreducibility of $G$ was not required.  It is needed only
for adding bars.}

If we add a bar, there is also a unique stressed subgraph.  This will
be all of $G$, again by the result of Streinu-Theran\cite{ST10},
unless both endpoints of the bar are in a common rigid subgraph.
That was ruled out by assuming that $G$ is irreducible.
\end{proof}

\section{Conclusions}
In this paper we have demonstrated boundary conditions for locally isostatic networks
that incorporate the right number of constraints at the surface so that the whole
network is isostatic.  These boundary conditions should be useful in numerical
simulations which involve finite pieces of locally isostatic networks.
The boundary can be quite complex and involve both an external boundary with
internal holes.

Our derivation of the new boundary conditions is based on a structural
characterization of graphs which capture the combinatorics of
silica bilayers.  This shows that the degrees of freedom are associated
with unpinned triangle corners on the boundary.  We then present two methods
to completely immobilize a triangle ring network: by attaching the boundary
to a wire rigidly attached to the plane; and by completely immobilizing
alternate vertices on the boundary.  To render a triangle ring network
isostatic, we also have two methods: adding bars between adjacent boundary
vertices in cyclic order; and attaching alternating boundary vertices to
an auxiliary graph that functions as a rigid frame.

Although our definition of a triangle ring network is most easily visualized
when $G$ is planar and $C$ is the outer face, the combinatorial setup
is quite a bit more general.  The natural setting for networks with
holes is to assume planarity, and then that all the degree $2$ vertices
are on disjoint facial cycles in $G$.  The key thing to note is that
the  cycle $C$ in our definition does not need to be facial for
\theoref{thumbtack-pinning}. For example, in \figref{Figure7},
$C$ goes around the boundary of an interior face that contains degree $2$
vertices.  In general, the existence of an appropriate cycle $C$ is a
non-trivial question, as indicated by \figref{Figure7} (See also Figure
\ref{fig:Figure5ab} for other examples of complex anchored boundary
conditions).

What is perhaps more striking is that \theoref{slider-pinning} still applies
whether or not such a $C$ exists, provided faces in $G$ defining the holes
in the sample are disjoint from the boundary and each other.

In applying anchored boundary conditions, it is important that the complete
boundary has an even number of unpinned sites, which can include internal holes,
which must then be connected using the green lines shown in the various figures.
This gives a practical way of setting up calculations with anchored boundary
conditions in samples with complex geometries and missing areas.

\begin{acknowledgements}
Support by the Finnish Academy (AKA) Project COALESCE is acknowledged by LT.
We thank Mark Wilson and Bryan Chen for many useful discussions and comments.
This work was initiated at the AIM workshop on configuration spaces,
and we thank AIM for its hospitality.
\end{acknowledgements}

\appendix

\section{Implementation details}
Algorithms \ref{alg:slider-pinning}--\ref{alg:bars-2}
in this appendix give a procedural
description of the four boundary conditions discussed
in this paper, and make clear the subtle differences
between them.
All of the algorithms in this appendix take as input a
finite part of a locally isostatic network and output
a globally isostatic one that is appropriate for
further study.  Which of the boundary conditions is most
appropriate will depend on the intended application.

Before describing the algorithms, we give
more detail on how to encode a triangle ring network and
the associated set of first-order geometric constraints.

\subsection{Encodings}
Computationally, it is convenient to work not only with the
body graph $G$, as in the main body of the paper,
but also with its \emph{line graph} $G^*$ that has as its vertices the triangle
corners and edges the triangle sides.  We denote by $n$ and $m$
the number of vertices and edges in $G$ and $n^*$ and $m^*$ the
same quantities for $G^*$.  Vertices in $G$ are denoted by $v, w,\ldots$
and vertices in $G^*$ by $v^*, w^*,\ldots$.  We assume that there is a
constant-time mapping $\tau : V(G)\to V(G^*)^3$ that maps
each vertex $v$ of $G$ to the associated triangle $\{v^*_v,w^*_v,x^*_v\}$
in $G^*$.  For each boundary vertex $v$ of $G$, $\tau(v)$ will have
a unique degree $2$ vertex, which we denote by $T(v)$.

Experimentally, $G^*$ will always be immediately visible.  It is also
computable in time $O(n)$ from $G$.  If $G$ is planar with given
facial structure\footnote{For example, given as a doubly-connected edge list.  See, e.g.,
Section 2.2 of the textbook by de Berg, et al.\cite{dBCvkO08}}, then $G^*$ also has
an natural planar embedding, and vice-versa.  Further, if $G^*$ contains no pair of facial
triangles with a common edge, then $G$ is determined by $G^*$.  This is the case in
all of our examples.

We also assume that we have access to the coordinates
of the vertices of $G^*$.  We denote these by
$p(v^*) = (x_{v^*},y_{v^*})$ for each vertex $v^*$ of $G^*$ and call $p$
a \emph{placement}.

\subsection{First-order geometric constraints}
The allowed first-order motions $\pdot$ of a triangle ring network
$G$ satisfy the system
\begin{eqnarray}\eqlab{inf-bars}
\iprod{p(v^*) - p(w^*)}{\pdot(v^*) - \pdot(w^*)} = 0 & \qquad\text{for all edges $v^*w^*\in E(G^*).$}
\end{eqnarray}
We assume that $p$ maximizes the rank of \eqref{inf-bars}, which happens for almost all
choices of $p$.  By the theorems in this paper, this rank is equal to $m^*$ when $G$ is a
triangle ring network.

Now identify a set $S^*\subset V(G^*)$ of vertices to which we will add one slider
constraint.  Assign a vector $s(v^*) = (a_{v^*},b_{v^*})$ to each $v^*\in S^*$.  The
slider constraints on the first-order motions are:
\begin{eqnarray}\eqlab{inf-sliders}
\iprod{s(v^*)}{\pdot(v^*)} = 0 & \qquad \text{for all $v^*\in S^*.$}
\end{eqnarray}
To guarantee that the combined system \eqref{inf-bars}--\eqref{inf-sliders}
achieves its maximum rank ($2n^*$ for our
sliding boundary condition), it is sufficient to pick each $s(v^*)$ uniformly
at random from the unit circle.\footnote{See the appendix of Király et al.\cite{KTT15} for a
detailed justification of this and similar statements relating genericity and random
sampling.}

\subsection{Implementing slider-pinning}
\algref{slider-pinning} shows how to implement the sliding boundary condition of
\theoref{slider-pinning}.
\begin{algorithm}[H]
\caption[Sliding boundary conditions]{Sliding boundary conditions.\\
\textit{Input:} Triangle ring network $G$, line graph $G^*.$\\
\textit{Output:} Slider constraints implemementing the sliding boundary condition. \alglab{slider-pinning}}
\begin{algorithmic}[1]
\item[1:] Initialize $S^*$ to the empty set.
\item[2:] For each boundary vertex $v$ of $G$, add $T(v)$ to $S^*$.
\item[3:] For each $v^*$ in $S^*$, generate a random vector $s(v^*)$ on the unit circle,
and use it to create a slider constraint of the form \eqref{inf-sliders}.
\end{algorithmic}
\end{algorithm}

\subsection{Implementing immobilized vertices}
To implement \theoref{thumbtack-pinning}, we could put two independent sliders at vertices of
$G^*$.  However, it is simpler to regard \eqref{inf-bars} as a matrix and then discard
the columns corresponding to immobilized vertices.
\begin{algorithm}[H]
\caption[Pinned boundary conditions]{Pinned boundary conditions.\\
\textit{Input:} Triangle ring network $G$ with an even number of boundary vertices,
line graph $G^*$, boundary cycle $C$.\\
\textit{Output:} Linear constraints pinning \eqref{inf-bars}. \alglab{thumbtack-pinning}}
\begin{algorithmic}[1]
\item[1:] Let $M$ be the matrix of the system \eqref{inf-bars}.
\item[2:] Pick a vertex $v_0$ on $C$.  Set $v = v_0$.  Set $b=1$.
\item[3:] If $b=1$, discard the columns in $M$ corresponding to $T(v)$,
and set $b=0$.  Otherwise set $b=1$.  Replace $v$ with its successor on $C$.
\item[4:] If $v = v_0$, output $M$.  Otherwise, go to step 3.
\end{algorithmic}
\end{algorithm}
Observe that the loop implemented in steps $2$--$4$ shows how to obtain the free corners of the
boundary vertices of $G$.

\subsection{Implementing anchoring with additional bars}
Anchoring with additional bars amounts to adding edges to $G^*$.  Thus, we describe
them graph theoretically only.  If the geometric constraints are desired, simply
use the new graph to write down \eqref{inf-bars}.
\begin{algorithm}[H]
\caption[Anchoring with bars I]{Anchoring with bars I.\\
\textit{Input:} Triangle ring network $G$,
line graph $G^*$, boundary cycle $C$.\\
\textit{Output:} An isostatic graph containing $G^*.$ \alglab{bars-1}}
\begin{algorithmic}[1]
\item[1:] Enumerate the boundary vertices $v_1,\ldots,v_t$ of $G$, ordered along $C$.
\item[2:] Set $H=G^*$.
\item[3:] Add edges $T(v_1)T(v_2),\ldots, T(v_3)T(v_2)$ to $H$.
\item[4:] Output $H$.
\end{algorithmic}
\end{algorithm}
The left panel of \figref{Figure8ab} takes the graph $H$ from \theoref{bars-2} to
be a ``zig-zag triangulation'' of a polygon, which is easily seen to be isostatic.
This next algorithm gives the implementation of \theoref{bars-2} using this
choice.
\begin{algorithm}[H]
\caption[Anchoring with bars I]{Anchoring with bars II.\\
\textit{Input:} Triangle ring network $G$ with an even number of boundary vertices,
line graph $G^*$, boundary cycle $C$.\\
\textit{Output:} An isostatic graph containing $G^*.$ \alglab{bars-2}}
\begin{algorithmic}[1]
\item[1:] Enumerate alternating boundary vertices $v_1,v_3,\ldots,v_t-1$ of $G$, ordered along $C$.
\item[2:] Set $H=G^*$.
\item[3:] Add edges $T(v_1)T(v_3),T(v_3)T(v_5),T(v_1)T(v_5)$ to $H$.
\item[4:] For $i=7,9,\ldots,t-1$, add the edges $T(v_{i-2})T(v_i),T(v_{i-4})T(v_i)$ to $H$.
\item[5:] Output $H$.
\end{algorithmic}
\end{algorithm}
\end{document}